\documentclass{article}

\usepackage[latin1]{inputenc}
\usepackage{graphicx}
\usepackage{amssymb,amsfonts,amsmath}
\usepackage{html}

\newtheorem{theorem}{Theorem}

\newtheorem{corollary}[theorem]{Corollary}

\newtheorem{lemma}[theorem]{Lemma}

\newenvironment{proof}[1][Proof]{\noindent\textbf{#1.} }{\newline \hspace*{\textwidth}\hspace*{-0,4cm} \rule{0.5em}{0.5em} \vspace{0,2cm}}

\begin{document}

\title{A wavelet-based tool for studying non-periodicity}
\author{R. Ben\'{\i}tez$^1$, V. J. Bol\'os$^2$, M. E. Ram\'{\i}rez$^3$ \\
\\
{\small $^1$ Departamento de Matem\'aticas,}\\
{\small Centro Universitario de Plasencia, Universidad Extremadura.}\\
{\small Avda. Virgen del Puerto 2, 10600 Plasencia, Spain.}\\
{\small e-mail\textup{: \texttt{rbenitez@unex.es}}} \\
\\
{\small $^2$ Departamento de Matem\'aticas para la Econom\'{\i}a y la Empresa,}\\
{\small Facultad de Econom\'{\i}a, Universidad de Valencia.}\\
{\small Avda. Tarongers s/n, 46071 Valencia, Spain.}\\
{\small e-mail\textup{: \texttt{vicente.bolos@uv.es}}} \\
\\
{\small $^3$ GMV A\& D, Spain.}\\
{\small e-mail\textup{: \texttt{mramirez@gmv.com}}}}
\date{June 2010}
\maketitle

\begin{abstract} This paper presents a new numerical approach to the
  study of non-periodicity in signals, which can complement the
  maximal Lyapunov exponent method for determining chaos transitions of a
  given dynamical system. The proposed technique is based on the
  continuous wavelet transform and the wavelet multiresolution analysis.
  A new parameter, the \textit{scale index}, is introduced and interpreted as
  a measure of the degree of the signal's non-periodicity. This methodology
  is successfully applied to three classical dynamical systems: the
  Bonhoeffer-van der Pol oscillator, the logistic map, and the Henon
  map.
\end{abstract}


\maketitle

\section{Introduction}
\label{introduction}

In the study of chaotic dynamical systems it is quite common to have bifurcation diagrams that represent, for each value of one
or more parameters, the number of periodic orbits of the system. The
determination of the parameter values for which the system becomes
chaotic is a classical problem within the theory of dynamical systems
\cite{Guck83}.

Although there is no universally accepted definition of chaos,
usually, a bounded signal is considered chaotic if (see
\cite{Strogatz94})

\begin{itemize}
\item[](a) it shows sensitive dependence on the initial conditions,
  and
\item[](b1) it is non-periodic, or
 \item[](b2) it does not
\textit{converge} to a periodic orbit.
\end{itemize}

Usually, chaos transitions in bifurcation diagrams are numerically
detected by means of the Maximal Lyapunov Exponent (MLE). Roughly
speaking, Lyapunov exponents characterize the rate of separation of
initially nearby orbits and a system is thus considered chaotic if the
MLE is positive. Therefore, the MLE technique is one that focuses on the
sensitivity to initial conditions, in other words, on criterion (a).

As to criteria (b1) and (b2), Fourier analysis can be used in order to
study non-periodicity. However chaotic signals may be
highly non-stationary, which makes wavelets more suitable
\cite{Chandre03}.  Moreover, compactly supported wavelets are a useful
tool in the analysis of non-periodicity in compactly supported
signals, as we will see in Corollary \ref{result}.

Wavelet theory is a quite recent area of mathematical research that
has been applied to a wide range of physical and engineering problems
(see, for instance \cite{Donald00,Tony05} for classical applications
to image processing and to time series analysis and
\cite{Arneodo98,Polig03,Panigr06} for examples of non-standard
applications to DNA sequences, astronomy and climatology). In
particular, the wavelet decomposition of a signal has been proved to
be a very useful tool in the study of chaotic systems. Indeed,
wavelets have been successfully used in the analysis of the chaotic
regimes of the Duffing oscillator, focusing on the detection of
periodicities within chaotic signals, and on chaos numerical control
(see \cite{Per92,Lak06}).

In this paper we present a method for studying non-periodicity that can
complement the MLE for determining chaos transitions of a given dynamical
system which can be either discrete or continuous. This method is
based on the Continuous Wavelet Transform (CWT) and the wavelet
Multiresolution Analysis (MRA) of a signal.  In particular, we compute
the ratio of the scalogram value at the dominant scale (i.e.  the
scale where the maximum is reached) to the value at the least
significant scale (i.e. the scale where the scalogram takes its
minimum value after that maximum is reached). This quotient determines
the \textit{scale index} which is strictly positive when the signal is
non-periodic, and can be interpreted as a measure of the degree of
non-periodicity.

The paper is organized as follows: Section \ref{sec:2} is devoted to
the establishing of the background and the main results of the
wavelet theory used later in the paper. In Section \ref{secindex} we
give an overview of the scalogram and define the scale index. Finally,
Section \ref{sec:4} illustrates the value of the method by
showing the successfull detection of the chaos transitions for three classical
systems: the Bonhoeffer-van der Pol (BvP) oscillator, the
logistic map and the Henon map.

\section{Wavelet analysis of time series}
\label{sec:2}

In this section we introduce the wavelet tools and results needed for
defining the \textit{scale index}.

\subsection{Continuous wavelet transform and scalogram}

Wavelet theory is based on the existence of two special functions
$\phi $ and $\psi $, known as the \textit{scaling} and \textit{wavelet
  functions} respectively \cite{mallat1999}.

A \textit{wavelet function} (or wavelet, for short) is a function
$\psi \in L^2\left( \mathbb{R}\right) $ with zero average (i.e.  $\int
_{\mathbb{R}} \psi =0$), with $\| \psi \| =1$, and \textit{centered}
in the neighborhood of $t=0$ \cite{mallat1999}. Moreover, we are going
to demand that $t\psi \left( t\right) \in L^1\left( \mathbb{R}\right)
$ in order to ensure that the continuous wavelet transform (\ref{eq:cwt}) is invertible
in some way.

Given a wavelet $\psi $, its dilated and translated dyadic version is given by
\begin{equation}
\label{eq:psijk}
\psi _{j,k}(t):=\frac{1}{\sqrt{2^j}}\psi \left( \frac{t-2^jk}{2^j}\right) ,
\end{equation}
where $j,k\in \mathbb{Z}$.  It is important
to construct wavelets such that the family of dyadic wavelets
$\{\psi_{j,k}\} _{j,k\in \mathbb{Z}}$ is an orthonormal basis of
$L^2\left( \mathbb{R}\right)$. These orthonormal bases are related to the
Multiresolution Analysis (MRA) of signals.

Scaling $\psi$ by a positive quantity $s$, and
translating it by $u\in\mathbb{R}$, we define a family of
\textit{time--frequency atoms}, $\psi_{u,s}$, as follows:
\begin{equation}
\label{eq:psius}
\psi _{u,s}(t):=\frac{1}{\sqrt{s}}\psi \left(
\frac{t-u}{s}\right) ,\qquad u\in \mathbb{R},\,\, s>0.
\end{equation}
Note that there is an abuse of notation in
expressions (\ref{eq:psius}) and (\ref{eq:psijk}).

Given $f\in L^2\left( \mathbb{R}\right) $, the \textit{continuous
wavelet transform} (CWT) of $f$ at time $u$ and scale $s$ is defined
as
\begin{equation}
Wf\left( u,s\right) :=\left<f,\psi _{u,s}\right> =\int
_{-\infty}^{+\infty}f(t)\psi ^* _{u,s}(t)\textrm{d}t,
\label{eq:cwt}
\end{equation}
and it provides the frequency component (or {\it details}) of $f$
corresponding to the scale $s$ and time location $t$.

The wavelet transform given in
(\ref{eq:cwt}), provides a time--frequency decomposition of $f$ in the
so called \textit{time--frequency plane} (see Figure
\ref{scaleoeoeogram_A076_A116}).

\begin{figure*}[htbp]
  \includegraphics[width=0.95\textwidth]{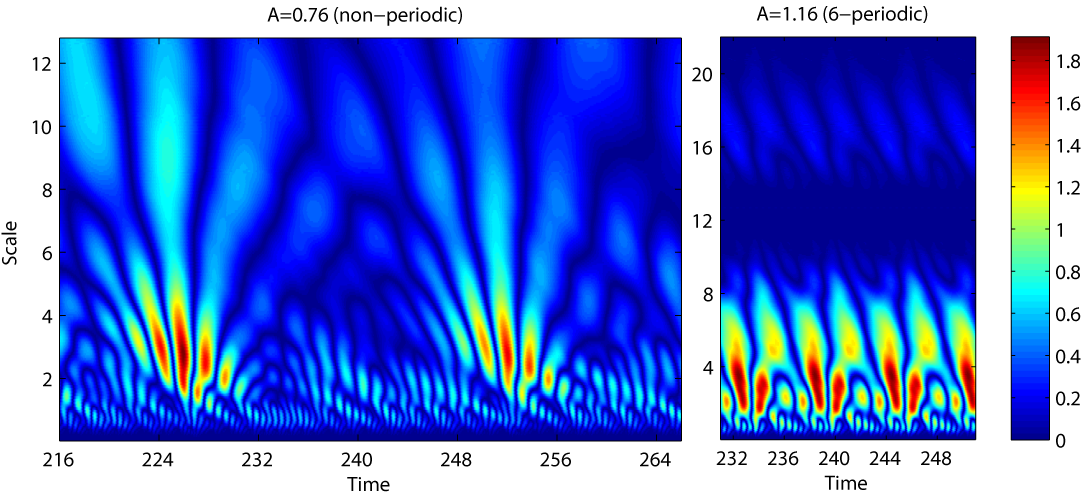}
  \caption{Time--frequency plane decomposition corresponding to the BvP
    solution with $A=0.76$ (left) and $A=1.16$ (right) using
    Daubechies (eight--wavelet and four--wavelet respectively) wavelet functions (see Section
    \ref{sec:bvp} for the definition of the BvP system). Each point in
    this 2D representation corresponds to the modulus of the wavelet
    coefficients of the CWT. Note that the wavelet coefficients of the
    CWT with $A=1.16$ vanish at scale 12 (i.e. twice its period) at
    any time, as we will prove in Theorem \ref{teo:1}.}
\label{scaleoeoeogram_A076_A116}
\end{figure*}

The \textit{scalogram} of $f$, $\mathcal{S}$, is defined as follows:
\begin{equation*}
  \mathcal{S}\left( s\right) :=\| Wf\left( u,s\right) \| =\left(
    \int _{-\infty}^{+\infty} | Wf\left( u,s\right)
    |^2\textrm{d}u\right) ^{\frac{1}{2}}.
\end{equation*}
$\mathcal{S}\left( s\right)$ is the \textit{energy} of the continuous
wavelet transform of $f$ at  scale $s$. Obviously, $\mathcal{S}(s)\geq
0$ for all scale $s$, and if $\mathcal{S}(s)>0$ we will say that the
signal $f$ has details at scale $s$. Thus, the scalogram is a useful
tool for studying a signal, since it allows the detection of its most
representative scales (or frequencies), that is, the scales that
mostly contribute to the total energy of the signal.

\subsection{Analysis of compactly supported discrete signals}

In practice, to make a signal $f$ suitable for a numerical study, we
have to
\begin{itemize}
\item[(i)] consider that it is defined over a finite time interval
  $I=\left[ a,b\right] $, and
\item[(ii)] sample it to get a discrete set of data.
\end{itemize}

Regarding the first point, boundary problems arise if the support of
$\psi _{u,s}$ overlaps $t=a$ or $t=b$. There are several methods for
avoiding these problems, like using \textit{periodic wavelets},
\textit{folded wavelets} or \textit{boundary wavelets} (see
\cite{mallat1999}); however, these methods either produce large
amplitude coefficients at the boundary or complicate the
calculations. So, if the wavelet function $\psi$ is compactly
supported and the interval $I$ is big enough, the simplest solution is
to study only those wavelet coefficients that are not affected by
boundary effects.

Taking into account the considerations mentioned above, the
\textit{inner scalogram} of $f$ at a scale $s$ is defined by
\begin{equation*}
  \mathcal{S}^{\textrm{inner}}\left( s\right) :=\| Wf\left( s,u\right) \| _{J(s)}
  =\left( \int _{c(s)}^{d(s)} | Wf\left( s,u\right) |^2
    \textrm{d}u\right) ^{\frac{1}{2}},
\end{equation*}
where $J(s)=\left[ c(s),d(s) \right]\subseteq I $ is the maximal
subinterval in $I$ for which the support of $\psi _{u,s}$ is included
in $I$ for all $u\in J(s)$. Obviously, the length of $I$ must be big
enough for $J(s)$ not to be empty or too small, i.e.  $b-a\gg sl$,
where $l$ is the length of the support of $\psi $.

Since the length of $J(s)$ depends on the scale $s$, the values of the
inner scalogram at different scales cannot be compared. To avoid this
problem, we can \textit{normalize} the inner scalogram:
\begin{equation*}
  \overline{\mathcal{S}}^{\textrm{inner}}\left( s\right)
  =\frac{\mathcal{S}^{\textrm{inner}}\left( s\right) }{(d(s)-c(s))^{\frac{1}{2}}}.
\end{equation*}

With respect to the sampling of the signal, any discrete signal can be
analyzed in a \textit{continuous way} using a piecewise constant
interpolation. In this way, the CWT provides a scalogram with a better
resolution than the Discrete Wavelet Transform (DWT), that considers
dyadic levels instead of continuous scales (see \cite{mallat1999}).

\section{The scale index}
\label{secindex}


In this section we introduce a new parameter, the \textit{scale
  index}, that will give us information about the degree of
non-periodicity of a signal. To this end we will first state some
results for the wavelet analysis of periodic functions (for further
reading please refer to \cite{mallat1999} and references therein).

If $f:\mathbb{R}\rightarrow \mathbb{C}$ is a $T$-periodic function in
$L^2\left( \left[ 0,T\right] \right) $, and $\psi $ is a compactly
supported wavelet, then $Wf\left( u,s\right) $ is well-defined for
$u\in \mathbb{R}$ and $s\in \mathbb{R}^+$, although $f$ is not in
$L^2\left( \mathbb{R}\right) $.

The next theorem gives us a criterion for distinguishing between periodic and
non-periodic signals. It ensures that if a signal $f$ has details at
every scale (i.e. the scalogram of $f$ does not vanish at any scale),
then it is non-periodic.

\begin{theorem}
\label{teo:1}
Let $f:\mathbb{R}\rightarrow \mathbb{C}$ be a $T$-periodic function in
$L^2\left( \left[ 0,T\right] \right) $, and let $\psi $ be a compactly
supported wavelet. Then $Wf\left( u,2T\right) =0$ for all $u\in
\mathbb{R}$.
\end{theorem}

For a detailed proof see Appendix \ref{sec:app}. From this result we
obtain the following corollary.

\begin{corollary}
\label{result}
  Let $f:I=\left[ a,b\right] \rightarrow \mathbb{C}$ a $T$-periodic
  function in $L^2\left( \left[ a,a+T\right] \right) $. If $\psi $ is
  a compactly supported wavelet, then the (normalized) inner scalogram
  of $f$ at scale $2T$ is zero.
\end{corollary}

\begin{figure*}[htbp]
\centering
\includegraphics[width=0.95\textwidth]{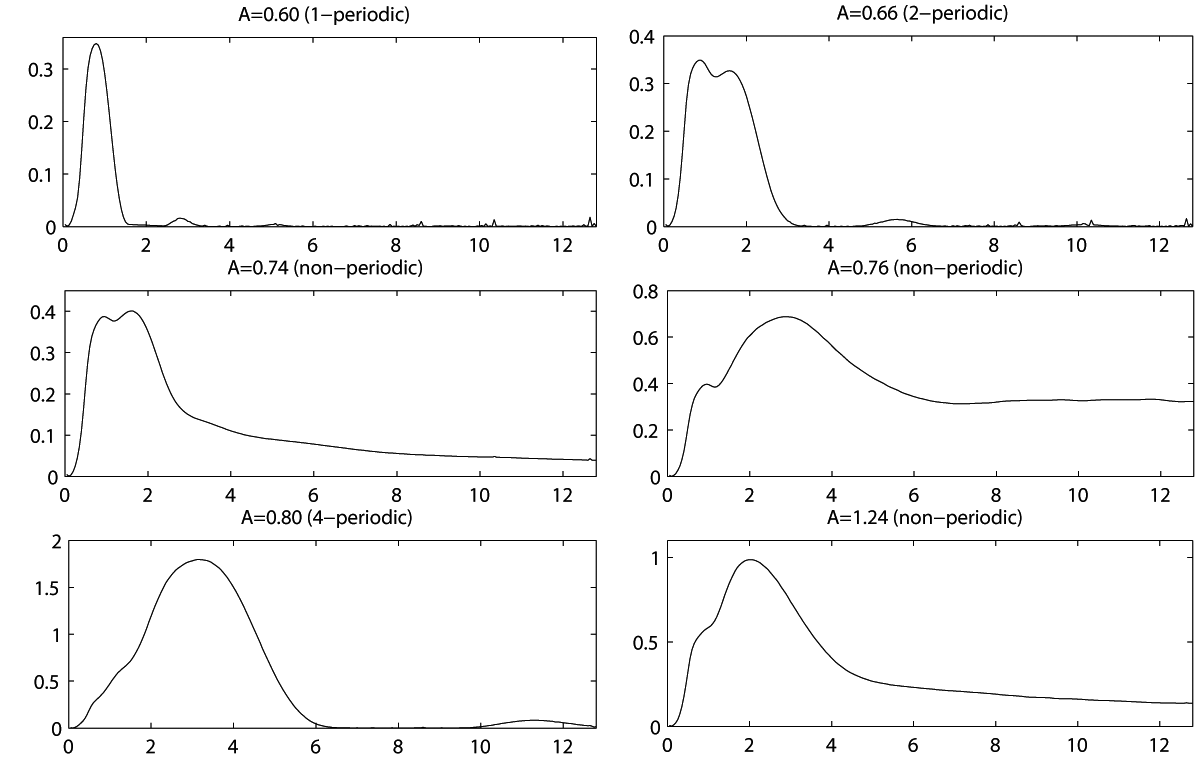}
\caption{Normalized inner scalograms for certain solutions of the BvP
  system (Section \ref{sec:bvp}), from $t=20$ to $t=400$ ($\Delta
  t=0.05$), for different values of $A$, the scale parameter $s$
  running from $s_0=0.05$ to $s_1=12.8$, with $\Delta s=0.05$, and
  using Daubechies eight--wavelet function. It is observed how the
  scalogram of $T$-periodic signals vanishes at $s=2T$.}
\label{bvp_scalogs}
\end{figure*}

These results constitute a valuable tool for detecting periodic and
non-periodic signals, because a signal with details at every scale
must be non-periodic (see Figure \ref{bvp_scalogs}). Note that in order to
detect numerically wether a signal \textit{tends to be periodic}, we have
to analyze its scalogram throughout a relatively wide time range.

Moreover, since the scalogram of a $T$-periodic signal vanishes at all
$2kT$ scales (for all $k\in \mathbb{N}$), it is sufficient to analyze only scales greater than a
fundamental scale $s_0$. Thus, a signal which has details at an
arbitrarily large scale is non-periodic.

In practice, we shall only study the scalogram on a finite interval
$[s_0,s_1]$. The most representative scale of a signal $f$ will be the
scale $s_{\textrm{max}}$ for which the scalogram reaches its maximum
value. If the scalogram $\mathcal{S}(s)$ never becomes too small
compared to $\mathcal{S}(s_{\textrm{max}})$ for $s>s_{\textrm{max}}$, then
the signal is ``numerically non-periodic'' in $[s_0,s_1]$.

Taking into account these considerations, we will define the
\textit{scale index} of $f$ in the scale interval $[s_0,s_1]$ as the
quotient
\[
i_{\textrm{scale}}:=\frac{\mathcal{S}(s_{\textrm{min}})}{\mathcal{S}(s_{\textrm{max}})},
\]
where $s_{\textrm{max}}$ is the smallest scale such that
  $\mathcal{S}(s)\leq \mathcal{S}(s_{\textrm{max}})$ for all $s\in
  [s_0,s_1]$, and $s_{\textrm{min}}$ the smallest scale such that
  $\mathcal{S}(s_{\textrm{min}})\leq \mathcal{S}(s)$ for all $s\in
  [s_{max},s_1]$. Note that for compactly supported signals only the
normalized inner scalogram will be considered.

From its definition, the scale index $i_{\textrm{scale}}$ is such that
$0\leq i_{\textrm{scale}}\leq 1$ and it can be interpreted as a
measure of the degree of non-periodicity of the signal: the scale
index will be zero (or numerically close to zero) for periodic signals
and close to one for highly non-periodic signals.

The selection of the scale interval $[s_0,s_1]$ is an important issue
in the scalogram analysis. Since the non-periodic character of a
signal is given by its behavior at large scales, there is no need for
$s_0$ to be very small. In general, we can choose $s_0$ such that
$s_{\textrm{max}}=s_0+\epsilon$ where $\epsilon$ is positive and close
to zero.

\begin{figure*}[htbp]
\centering
 \includegraphics[width=0.95\textwidth]{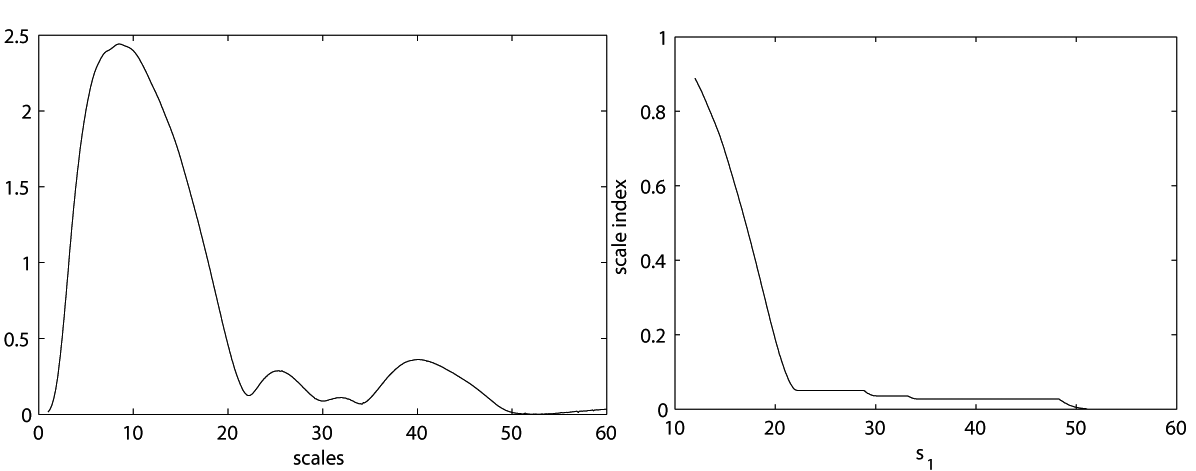}
 \caption{Left: normalized inner scalogram of the almost-periodic
   function $\textrm{sin}\left( t\right)+\textrm{sin}\left(
     t/\sqrt{2}\right) +\textrm{sin}\left( t/\sqrt{5}\right) $, from
   $t=0$ to $t=800$ ($\Delta t=0.1$), using the Daubechies four--wavelet
   function. Right: the scale index tends to zero as we increase the scale
   $s_1$.}
\label{fig:almostper}
\end{figure*}

On the other hand, $s_1$ should be large enough for detecting
periodicities. For example, if we have an almost-periodic function $f$
defined on $\mathbb{R}$, then for any given
$\epsilon>0$ there exists an \textit{almost-period} $T\left(
  \epsilon\right) $ such that
\[
\left| f\left( t+T\right) -f\left( t\right) \right| <\epsilon
\]
for all $t\in \mathbb{R}$ (see \cite{almostper}). Hence, if we choose $\epsilon $ numerically
close to zero, there is a large enough value of $T$ for which the function
is ``numerically $T$-periodic'', and the scale index will be close to
zero if $s_1$ is greater than $2T$. So, for an almost-periodic
function, the scale index tends to zero as we increase $s_1$ (see
Figure \ref{fig:almostper}). But as $s_1$ increases, so does the
computational cost. In fact, the larger $s_1$ is, the wider the time
span should be where the signal is analyzed, in order to maintain the
accuracy of the normalized inner scalogram.

Scales $s_{\textrm{min}}$ and $s_{\textrm{max}}$ determine the
pattern that the scalogram follows (see Figure \ref{bvp_smax_smin_z1}).
For example, in non-periodic signals  $s_{\textrm{min}}$ can be
regarded as the ``least non-periodic scale''. Moreover,
if $s_{\textrm{min}}\simeq s_1$, then the
scalogram decreases at large scales and $s_1$ should be increased
in order to distinguish between a non-periodic signal and a periodic
signal with a very large period.

\begin{figure*}[htbp]
\centering
 \includegraphics[width=0.85\textwidth]{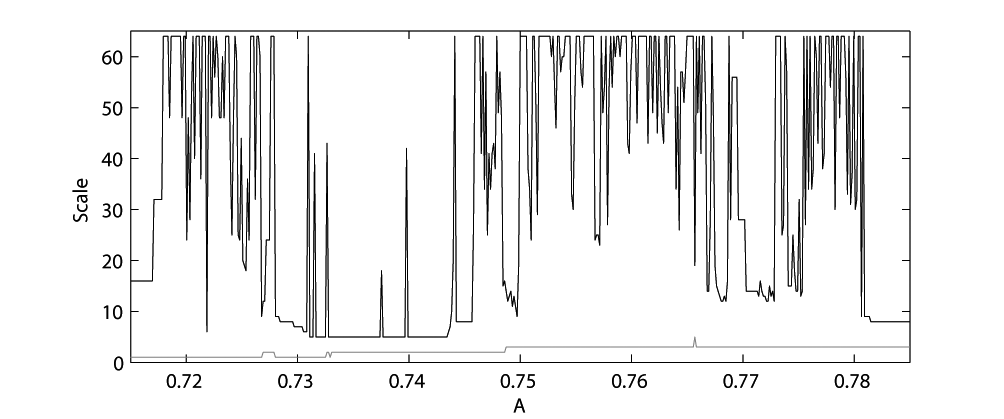}
 \caption{Values of $s_{\textrm{max}}$ (grey) and $s_{\textrm{min}}$
   (black) for orbits of the BvP oscillator.}
\label{bvp_smax_smin_z1}
\end{figure*}

\section{Examples}
\label{sec:4}

In this section we illustrate how the scale index $i_{\textrm{scale}}$
is used in order to detect and study non-periodic orbits of three classical
dynamical systems: the logistic map, the Henon map and the forced
Bonhoeffer-van der Pol oscillator.

The reason for choosing these dynamical systems as examples for
testing the validity of the scale index is mainly that they are
three well known chaotic dynamical systems, arising from different
research areas, and are mathematically very different. These systems
present typical bifurcation diagrams with chaotic and non-chaotic
regions. In order to show the effectiveness of the index
$i_{\textrm{scale}}$, we compare the bifurcation diagram,
the MLE, and $i_{\textrm{scale}}$. It will be shown that there is a
correspondence between the chaotic regions of the bifurcation diagram,
the regions where the MLE is positive, and the regions where
$i_{\textrm{scale}}$ is positive.

Figures \ref{bvp_bif_MLE_cimm} and \ref{logistic_henon_bif_MLE_cimm}
depict the comparison between the three methods mentioned above. The
signals were studied from $t_0 = 20$ to identify not only periodic
signals, but also signals that converge to a periodic one. For the
computation of the MLE, 1500 iterations were used. Integer scales
between $s_0 = 1$ and $s_1 = 64$ were considered in the computation of
$i_{\textrm{scale}}$. A scale was considered to have no details if the
scalogram at that scale takes a value below $\epsilon = 10^{-4}$.

\subsection{Continuous Dynamical System: The Bonhoeffer-van der Pol
  Oscillator}
\label{sec:bvp}

The Bonhoeffer-van der Pol oscillator (BvP) is the non-autonomous
planar system
\begin{equation*}
\left. \begin{array}{rcl}
x'&=&x-\displaystyle{\frac{x^3}{3}}-y+I(t) \\
y'&=&c(x+a-by)\\
\end{array} \right\} ,
\end{equation*}
being $a$, $b$, $c$ real parameters, and $I\left( t\right) $ an
external force. We shall consider a periodic force $I(t)=A\cos
\left( 2\pi t\right) $ and the specific values for the parameters
$a=0.7$, $b=0.8$, $c=0.1$. These values were considered in
\cite{Raja96} because of their physical and biological importance (see
\cite{Scot77}).

The classical analysis of the BvP system is focused on its Poincaré
map, defined by the flow of the system on $t=1$ (see \cite{Guck83}).
Plotting the first coordinate of the periodic fixed points of the
Poincaré map versus the parameter $A$ (amplitude of the external
force), a bifurcation diagram is obtained \cite{Wang89}.

Such diagrams present chaotic and non-chaotic zones. From a geometric
point of view, chaos transitions are related to homoclinic orbits
(creation or destruction of Smale horseshoes) between the invariant
manifolds of a saddle fixed point of the Poincar\'e map (see
\cite{Guck83}). Such a relationship is thoroughly described in a recent
work \cite{BenBol08}.

Figure \ref{bvp_bif_MLE_cimm} includes the analysis of the BvP
system. The parameter range has been split in two regions, $0.7\leq
A\leq 0.8$ and $1\leq A\leq 1.3$, which are the regions were chaotic
orbits are found.  Note the high level of agreement between the MLE and the
$i_{\textrm{scale}}$ index: the values of $A$ for which the MLE is
negative are also the values for which $i_{\textrm{scale}} \approx 0$.

\begin{figure*}[htbp]
\centering
\includegraphics[width=0.95\textwidth]{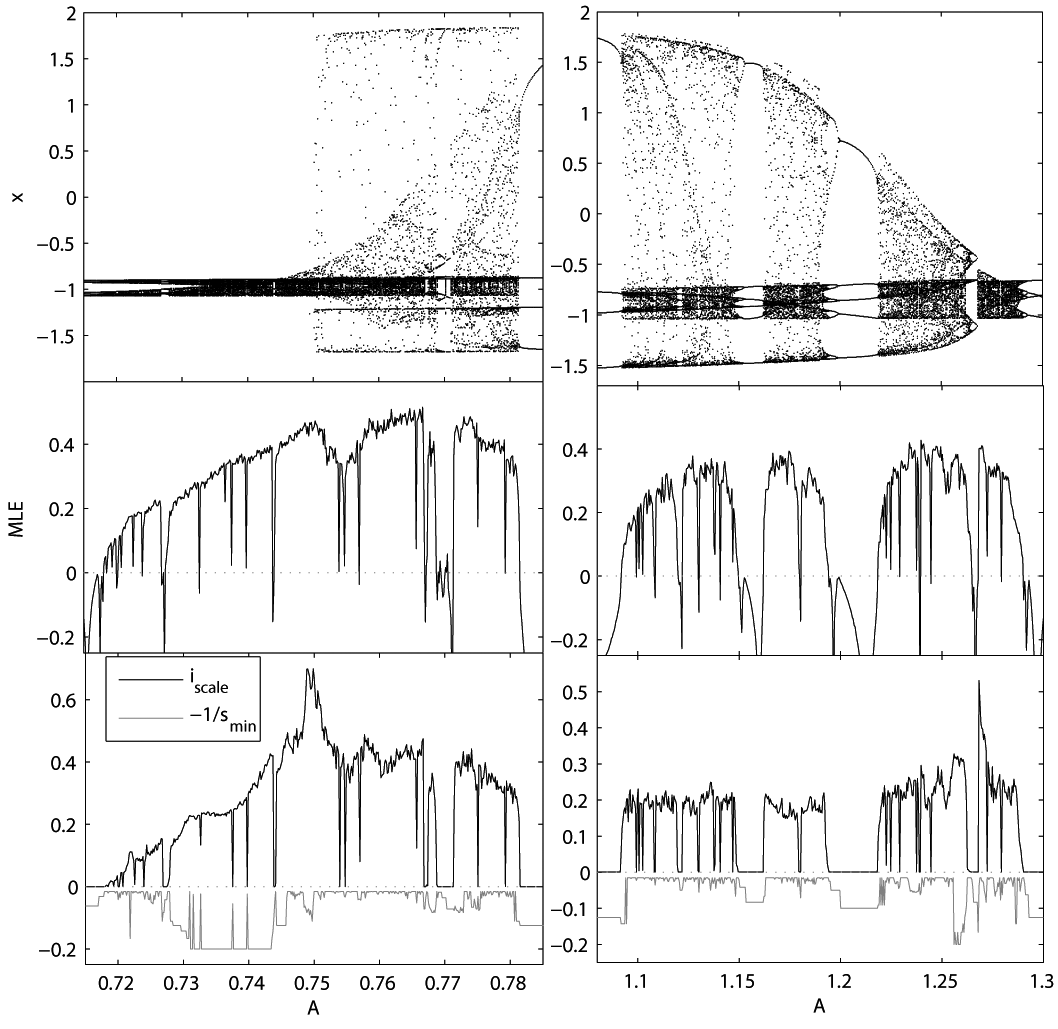}
\caption{Comparison between the bifurcation diagram, MLE and
  $i_{\textrm{scale}}$ (from top to bottom) for the
  BvP oscillator.}
\label{bvp_bif_MLE_cimm}
\end{figure*}

Also remarkable is the coincidence between a relative maximum in
the $i_{\textrm{scale}}$ and the well known ``sudden'' expansion of
the size of the attractor at $A\approx 0.748$ (see \cite{BenBol08},
\cite{Raja96}). Moreover, another relative maximum fits with a
``sudden'' contraction at $A\approx 1.27$.

\subsection{Discrete Dynamical Systems: The Logistic and the Henon
  maps}

The logistic map is the discrete dynamical system given by the
difference equation
\[
x_{t+1}=A x_t(1-x_t).
\]
This well known dynamical system, arising from population dynamics
theory, is a classical example of a simple polynomial map whose orbits
exhibit chaotic behavior for some values of the real parameter $A$.

On the other hand, the Henon map is the two-dimensional dynamical
system given by the quadratic map
\[
\left\{ \begin{array}{rcl}
x_{t+1} &=& 1-A x^2_t+y_t \\
y_{t+1} &=& b x_t \\
\end{array} \right.,
\]
with $A$ and $b$ real parameters. For the values $A = 1.4$ and $b =
0.3$, giving what is often called the \textit{canonical Henon map}, a strange
attractor is present. Fixing the value $b = 0.3$ and varying the
parameter $A$, the map may be chaotic or not. Through the
representation of the $x$-coordinate of the orbits versus the value of
the parameter $A$, a typical bifurcation diagram is obtained.

Figure \ref{logistic_henon_bif_MLE_cimm} shows the comparison between
the bifurcation diagram, the MLE and $i_{\textrm{scale}}$ for both the
logistic map (left) and the Henon map (right). As in the case of the
BvP oscillator, the agreement between the MLE and
$i_{\textrm{scale}}$ is clear.

It is also noticeable that the maximum values of the
$i_{\textrm{scale}}$ in these two dynamical systems are reached when
the main branches of their bifurcation diagrams overlap.

\begin{figure*}[htbp]
\centering
\includegraphics[width=0.95\textwidth]{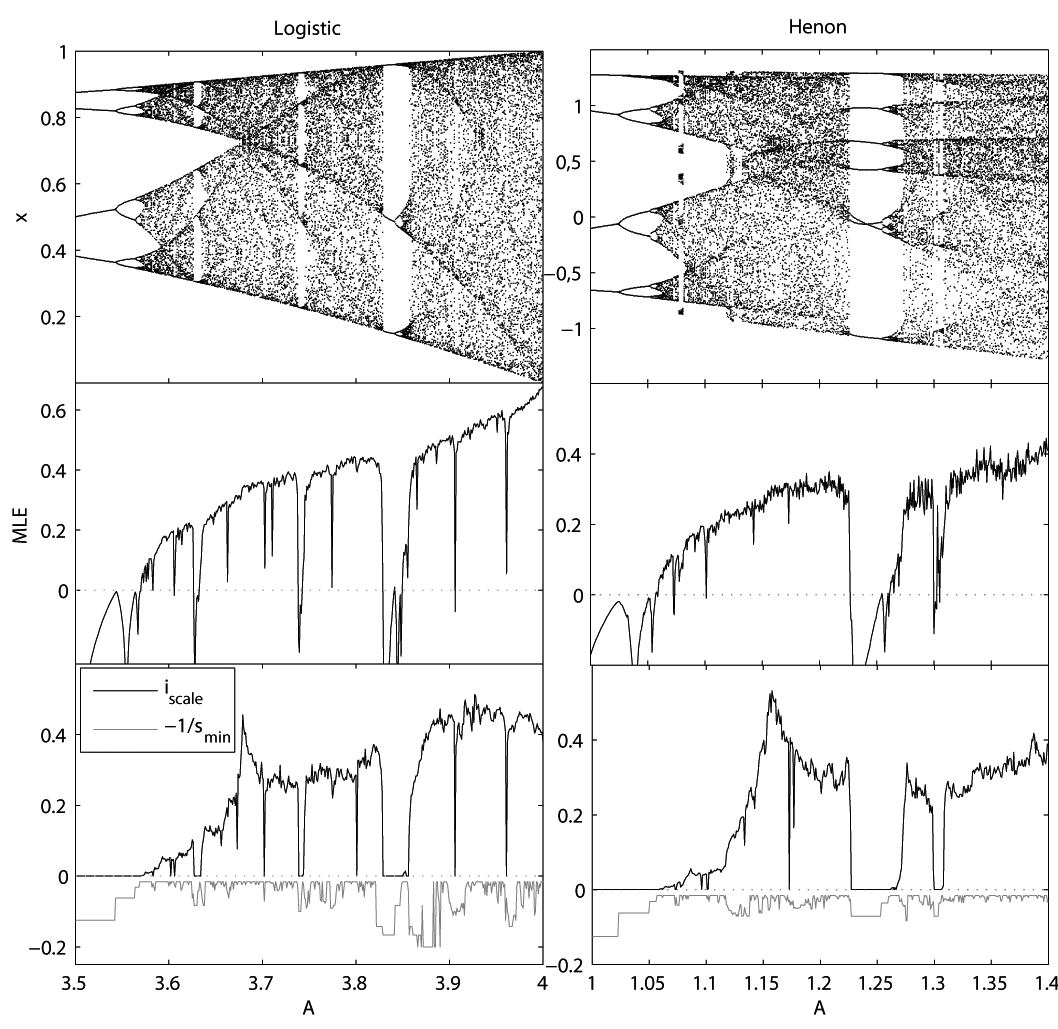}
\caption{Comparison between the bifurcation diagram, MLE and $i_{\textrm{scale}}$ (from top to
bottom) for the logistic map (left) and the Henon map (right).}
\label{logistic_henon_bif_MLE_cimm}
\end{figure*}

\section{Conclusions}

Wavelet analysis has proved to be a valuable tool in the study of
chaotic systems. In particular, scalogram analysis and the introduced
scale index are a good complement to the MLE, since the scale index
gives a measure of the degree of non-periodicity of the signal, while
the MLE gives information about the sensitivity to initial conditions.
Thus, the combination of the two methods gives us a
comprehensive description of a chaotic signal.

Since the MLE and the scale index focus on different characteristics of the
signals, the latter contributes to detecting effects that are not
detected by the MLE, for example, the sudden expansion of the size
of the attractor in the BvP system, and the overlapping of the main
branches of the bifurcation diagram in the logistic and Henon
maps. Moreover, there are regions in the logistic and Henon graphs
where the MLE is increasing but the $i_{\textrm{scale}}$ is
decreasing. This means that while the sensitivity to the initial
conditions is greater, on the other hand the signal is less
non-periodic.


Additionally, the study of the scale index does not require an analytical
expression for the signal. In those cases where an analytical
description of the dynamical system involved is not available
(e.g. experimental signals), although there are methods for estimating the
MLE (see \cite{Zeng91}), the scale index might be a useful
alternative.

These techniques can also be applied to any other discipline where the
analysis of time series is required, such as Earth sciences,
econometry, bio\-me\-di\-cine and any other one where non-linear behavior is
expected to occur. For instance, this opens a branch to the study of
colour noise presented in this kind of time series; it also could be
used to determine properties of seismic and volcanic events
(\cite{Ibanez03}); to detect chaos in non-linear economic time series
(\cite{Barnett92}) as a complement to
other methods already used, like those based on Lyapunov exponents and the power spectral
density; and in biomedicine, where wavelets have already been used to
analyze non-stationary cardiac signals (\cite{Faust04}).


\appendix
\section{Proof of Theorem \ref{teo:1}}
\label{sec:app}

Given $g:\mathbb{R} \rightarrow \mathbb{C} $ with compact support, its
periodization over $\left[ 0,1\right] $ is defined as
\[
g^{\textrm{per}}(t):=\sum _{k\in \mathbb{Z}} g(t+k),\qquad t\in \mathbb{R}.
\]
It is clear that, if $g_u(t):=g(t-u)$ is a translated version of $g$, then
\begin{equation}
\label{eq:guper}
\left( g_u\right) ^{\textrm{per}}(t)=g^{\textrm{per}}(t-u).
\end{equation}

Moreover, if $h\in L^2\left( \left[ 0,1\right] \right) $, then
\begin{equation}
\label{eq:hperg}
\langle h^{\textrm{per}},g\rangle =
\langle h,g^{\textrm{per}}\rangle _{\left[ 0,1\right] }.
\end{equation}

\begin{lemma}
\label{lema1}
The periodization $\psi ^{\textrm{per}}_{u,s}$ is zero almost
everywhere for $u=0$ and $s=2$.
\end{lemma}
\begin{proof}
  In this proof we are going to use the dyadic notation for the
  dilated and translated wavelets given by expression
  (\ref{eq:psijk}). So, we have to prove that $\psi
  ^{\textrm{per}}_{j,k}=0$ a.e. for $j=1$ and $k=0$.

  Let $h\in L^2\left( \left[ 0,1\right] \right) $.  Since $\left\{
    \psi ^{\textrm{per}}_{j,k}\right\} _{-\infty <j\leq 0,\, 0\leq
    k<2^{-j}} \cup \left\{ 1\right\} $ is an orthogonal basis of
  $L^2\left( \left[ 0,1\right] \right) $ and the periodization $\psi
  ^{\textrm{per}}_{1,0}$ is orthogonal to the family of functions
  $\left\{ \psi ^{\textrm{per}}_{j,k}\right\} _{-\infty <j\leq 0,\,
    0\leq k<2^{-j}}$ (see \cite[Thm. 7.16, Lem. 7.2]{mallat1999}), we
  have
\[
\langle h,\psi ^{\textrm{per}}_{1,0}\rangle _{\left[ 0,1\right] }
=M\langle 1,\psi ^{\textrm{per}}_{1,0}\rangle _{\left[ 0,1\right] }
=M\langle 1,\psi _{1,0}\rangle ,
\]
where $M=\langle h,1\rangle _{\left[ 0,1\right] }$.  Taking into
account that $\psi $ has zero average, we have that $\langle 1,\psi
_{1,0}\rangle =0$, and so
\[
\langle h,\psi ^{\textrm{per}}_{1,0}\rangle _{\left[ 0,1\right] }=0.
\]

In general, it can be proved that $\psi ^{\textrm{per}}_{j,k}=0$
a.e. for all $j\geq 1$ and for all $k\in \mathbb{Z}$.
\end{proof}

\begin{lemma}
\label{lema2}
The periodization $\psi ^{\textrm{per}}_{u,2}=0$ a.e. for all $u\in \mathbb{R}$.
\end{lemma}
\begin{proof}
  Since $\psi _{u,2}(t)=\psi _{0,2}(t-u)$ we have that $\psi
  ^{\textrm{per}}_{u,2}(t)=\psi ^{\textrm{per}}_{0,2}(t-u)=0$
  a.e. taking into account (\ref{eq:guper}) and Lemma \ref{lema1}.
\end{proof}

\textbf{Theorem \ref{teo:1}}. \textit{
  Let $f:\mathbb{R}\rightarrow \mathbb{C}$ be a $T$-periodic function
  in $L^2\left( \left[ 0,T\right] \right) $, and let $\psi $ be a
  compactly supported wavelet. Then $Wf\left( u,2T\right) =0$ for all
  $u\in \mathbb{R}$.}

\begin{proof}
  Without loss of generality we may assume that $T=1$. Then, it will
  be proved that $Wf\left( u,2\right) =0$ for all $u\in \mathbb{R}$:

  Let $h:=f|_{\left[ 0,1\right] }$. Then $f=h^{\textrm{per}}$ and we
  get
\[
Wf\left( u,2\right)=\langle f,\psi _{u,2}\rangle=\langle
h^{\textrm{per}},\psi _{u,2}\rangle = \langle h,\psi
^{\textrm{per}}_{u,2}\rangle _{\left[ 0,1\right] }=0,
\]
taking into account (\ref{eq:hperg}) and Lemma \ref{lema2}.
\end{proof}

\section*{Acknowledgments}
  We thank Carlos Fern\'andez Garc\'{\i}a, the referees, and the editor of the journal \textit{Computer and Mathematics with Applications} for their
  useful comments and suggestions that improved the paper.

\end{document}